\documentclass[a4paper]{article}

\newtheorem{theorem}{Theorem}
\newtheorem{lemma}{Lemma}

\newtheorem{Definition}{Definition}
\newtheorem{proof}{Proof}

\usepackage[draft,mode=multiuser]{fixme}
\fxuseenvlayout{colorsig}
\fxusetargetlayout{color}
\FXRegisterAuthor{if}{aif}{\color{blue}IF}
\usepackage[symbol]{footmisc}

\usepackage{url}
\usepackage{xspace}
\usepackage{mathtools}
\usepackage{amsfonts}
\usepackage{latexsym}
\usepackage{algorithm}
\usepackage{algpseudocode}

\newcommand{\ord}{\mathcal{O}}
\newcommand{\rp}{RePair\xspace}
\newcommand{\rpps}{RePair(PS)\xspace}
\newcommand{\mrrp}{MR-RePair\xspace}
\newcommand{\rlmrrp}{RL-MR-RePair\xspace}

\newcommand{\occ}[2]{\#\mathrm{occ}_{#1}({#2})}

\begin{document}

\title{\Large Practical Repetition-Aware Grammar Compression}
\author{
Isamu~Furuya$^{\dagger}$\\[0.5em]
{\small\begin{minipage}{\linewidth}\begin{center}\begin{tabular}{c}
	$^{\dagger}$ Graduate School of IST, Hokkaido University, Japan\\
	\url{furuya@ist.hokudai.ac.jp}\\
	\end{tabular}\end{center}\end{minipage}}
}

\date{}

\maketitle

\begin{abstract} \small\baselineskip=9pt 
	The goal of grammar compression is 
	to construct a small sized context free grammar which uniquely generates the input text data.
	Among grammar compression methods,
	\rp is known for its good practical compression performance.
	\mrrp was recently proposed as an improvement to \rp for 
	constructing small-sized context free grammar
	for repetitive text data.
	However, a compact encoding scheme has not been discussed for \mrrp.
	We propose a practical encoding method for \mrrp
	and show its effectiveness through comparative experiments.
	Moreover, we extend \mrrp to run-length context free grammar
	and design a novel variant for it called \rlmrrp.
	We experimentally demonstrate that a compression scheme consisting of \rlmrrp and the proposed encoding method
	show good performance on real repetitive datasets.
\end{abstract}

\section{Introduction}\label{sec:introduction}

Grammar compression is a method of lossless data compression 
that reduces the size of a given text 
by constructing a small-sized context free grammar (CFG) uniquely derived from the input text.
While the problem of generating the smallest such grammar is NP-hard~\cite{Charikar+:IEEE2005},
several approximation techniques have been proposed.
Among them, \rp~\cite{Larsson+:IEEE2000} is known as an off-line method 
that can practically achieve a high compression ratio~\cite{Claude+:TWEB2010,Conzalez+:CPM2007,Wan:phD2003}
despite its simple scheme.
Considerable study has been conducted on the subject of \rp, such as
extending it to an online algorithm~\cite{Masaki+:DCC2016},
practical working time/space improvements~\cite{Bille+:DCC2017,Sekine+:DCC2014},
applications to other fields~\cite{Claude+:TWEB2010,Lohrey+:INFSYST2013,Tabei+:KDD2016},
and theoretical analysis on the generated grammar sizes~\cite{Charikar+:IEEE2005,Navarro+:DCC2008,Ochoa+:IEEE2019}.

Recently, Furuya et al.~\cite{Furuya+:DCC2019}
proposed a variant of \rp called \mrrp,
which they claimed practically achieves more efficient compression than \rp
in terms of the size of constructed grammars for repetitive datasets.
Generating grammar of small sizes has considerable importance
since there are several algorithms and data structures that are applied to grammar compressed texts,
with running times that depend on the sizes of the grammar
~\cite{Gawrychowski:SPIRE2012,Goto+:JDA2013,Bille+:TCS2014,I+:TCS2015,Jez:TA2015}.
At the same time, 
it is also important to
encode these grammars as compact bit sequences
since compressed data are stored as bit sequences.
However, the authors in~\cite{Furuya+:DCC2019} did not discuss encoding methods for constructed grammar.
Related to \rp,
succinct encoding of straight-line program (SLP) was addressed in~\cite{Tabei+:CPM2013}
(note that grammars constructed by \rp are easily transformed to SLPs).
In addition,
Bille et al.~\cite{Bille+:DCC2017,Prezza:github} proposed a variant of \rp and an effective method for practically encoding it.
Without effective methods for encoding \mrrp,
the final bit sequence of its grammar may be larger than that of \rp,
even if the size of grammar for \mrrp is smaller than that of \rp.
In this work,
we propose a compact bit encoding method for \mrrp
and experimentally demonstrate its effectiveness and the proposed encoding method.

Run-length CFG (RLCFG) is an extension of CFG,
used by Je\.{z}~\cite{Jez:TCS2015} 
but formally introduced by Nishimoto et al.~\cite{Nishimoto+:MFCS2016}.
RLCFG improves the effectiveness of CFG compression in theory;
its theoretical properties were studied in~\cite{Bille+:JDA2018,Gagie+:SODA2018}.
However, to the best of our knowledge,
the experimental efficiency of the compression using RLCFG is undiscovered
since there are no compression algorithms on RLCFG scheme.
In this study,
we extend \mrrp and
propose a compression algorithm on RLCFG, called \rlmrrp.
We report that
experiments show that
\rlmrrp constructed smaller grammars for repetitive datasets
compared with both \rp and \mrrp.
Moreover,
we applied the proposed encoding method to \rlmrrp and experimentally show its effectiveness on
real repetitive datasets.

\noindent
{\bf Contributions:} The primary contributions of this study are as follows.
\begin{enumerate}
	\item
		We extended \mrrp using the RLCFG scheme and
		designed a compression algorithm, called \rlmrrp.
	\item 
		We proposed an encoding scheme for \mrrp and \rlmrrp.
	\item
		We implemented \rp, \mrrp and \rlmrrp and experimentally confirmed 
		that \rlmrrp produces smaller grammars than the others in nearly all instances.
		Moreover, 
		we implemented 8 encoding methods for \rp and
		6 encoding methods for \mrrp and \rlmrrp and
		showed the effectiveness of their compression
		through comparative experiments.
\end{enumerate}

The remainder of this paper is organized as follows.
In Section~\ref{sec:preliminaries}, 
we provide baseic notations and review some definitions and previous algorithms.
In Section~\ref{sec:grammar}, 
we define the \rlmrrp algorithm.
In addition, we describe its implementation and analyze the time/space complexity.
In Section~\ref{sec:encoding},
we introduce some encoding scheme for grammar compression
and present our bit encoding method.
In Section~\ref{sec:experiments},
we show experimental results.
Finally, we conclude the study in Section~\ref{sec:conclusion}.

\section{Preliminaries}\label{sec:preliminaries}

In this section, we provide some notations and definitions to be used in the following sections.
In addition, 
we review grammar compression and some basic encoding methods.

\subsection{Definitions and basic notation}

Let $\mathcal{A}$ denote an {\em alphabet}, an ordered finite set of symbols.
A {\em text} $T = t_1\cdots t_n$ is an element of $\mathcal{A}^{\ast}$,
and $|T| = n$ denotes its length.
For two natural numbers $i$ and $j$ such that $1 \le i \le j \le n$,
a {\em substring} of $T$ is $t_i \cdots t_j$, which can be denoted by $T[i..j]$.
If $i=j$; for simplicity, we denote it by $T[i]$ (this is the $i$-th symbol of $T$).

For text $T$ and $u$, $\occ{T}{u}$ denotes the {\em frequency} of $u$ in $T$,
representing the number of occurrence of $u$ in $T$.
Assume that $\occ{T}{u}\ge 1$,
then a {\em left} (or {\em right}) {\em extension} of $u$ is 
any substring of $T$ that has the form $wu$ (or $uw$), where $w \in \mathcal{A}^{+}$.
We say that $u$ is {\em left} (or {\em right}) {\em maximal}
if left (or right) extensions of $u$ occur strictly fewer times in $T$ than in $u$,
or $u$ has no left (or right) extension.
If $u$ is left and right maximal and $\occ{T}{u}\ge 2$, we call $u$ a {\em maximal repeat}.
Moreover, a {\em conditional maximal repeat} is the longest substring $u$ of a maximal repeat 
such that $u[1]\neq u[|u|]$. 
In this study, we consider only strings with a length greater than $1$ as maximal repeats.
For a maximal repeat $u$ with length 2 and $u[1]=u[2]$,
we regards $u$ itself as its conditional maximal repeat.

A {\em repetition} is a text that has the form $w^k$ with $w\in\mathcal{A}^{+}$ and $k\in\mathbb{N}^{+}$,
which means $k$ repetitions of $w$.
A {\em run} is a repetition that satisfies both of the following two conditions:
(i) $w\in\mathcal{A}$, and
(ii) any of its left and right extensions are not repetition,
or it has no left and right extension.

\subsection{Grammar compression}

A {\em context free grammar} (CFG) $G$ is
a 4-tuple $G = \{\Sigma, V, s, R\}$,
where $\Sigma$ is an ordered finite set of {\em terminals},
$V$ is an ordered finite set of {\em variables} such that $V$ and $\Sigma$ are disjoint,
$R$ is a finite set of binary relations called {\em production rules} (or {\em rules})
between $V$ and $(\Sigma \cup V)^{\ast}$,
and $s\in V$ is a special variable called the {\em start symbol}.
A production rule represents an expansion rule of a variable,
written in the form $v \rightarrow \alpha$,
where $v \in V$ and $\alpha \in (\Sigma \cup V)^{\ast}$
imply that $v$ expands to $\alpha$.
Here, the word {\em symbol} refers to any element of $(\Sigma \cup V)$;
we call an element of $(\Sigma \cup V)^{\ast}$ a text.
A {\em straight-line program (SLP)}~\cite{Karpinski+:NJC1997} is a canonical form of CFG 
such that all the length of the right-hand side of the rules is 2.

A {\em run-length context free grammar} (RLCFG)~\cite{Jez:TCS2015,Nishimoto+:MFCS2016} is an extension of CFG
by adding {\em run-length rules} to production rules, 
written in the form $v \rightarrow \alpha^{k}$
with $\alpha\in(\Sigma\cup V)$ and $k\ge 1$.
We call both CFG and RLCFG {\em grammar} and
say that a grammar {\em generates} a text $T$
if its start symbol $s$ expands to $T$ by recursively applying its production rules.

The {\em parse tree} of grammar is a rooted ordered tree 
with internal nodes labeled by variables and leaves labeled by terminals
such that 
the relation between internal node $v_i$ and its children corresponds to the rule $v_i\rightarrow\alpha_i$,
that is, if $\alpha_i = v_{i_1}, v_{i_2}, \cdots, v_{i_j}$ with $j = |\alpha_i|$,
children of $v_i$ are nodes labeled by $v_{i_1}, v_{i_2}, \cdots, v_{i_j}$ from left to right.
Note that the label sequence of the leaves of the parse tree represents the text generated by the grammar.

Given a text $T$, 
{\em grammar compression} is a method of lossless compression
that constructs a grammar $G$
such that $G$ generates $T$ uniquely.
For $G$,
we assume that
there is a unique rule $v_i\rightarrow \alpha_i$ for each variable $v_i\in V$,
where $\alpha_i$ is
either $\alpha_i = a$ ($a \in \Sigma$) 
or $\alpha_i = v_{j_1}v_{j_2}\cdots v_{j_m}$ ($i > j_k$ for all $1\le k \le m$)
or $\alpha_i = {v_j}^k$.

For each form of rule,
its size is defined as follows;
(i) for $v\rightarrow a$, its size is 1,
(ii) for $v\rightarrow v_{j_1}v_{j_2}\cdots v_{j_m}$, its size is $m$, and
(iii) for $v\rightarrow {v_j}^k$, its size is 3.
We estimate the effectiveness of compression by the size of a generated grammar,
which is the total size of its production rules.

\subsection{Previous algorithms}

{\em \rp} is an algorithm for grammar compression
proposed by Larsson and Moffat~\cite{Larsson+:IEEE2000}.
For a given text $T$, \rp constructs CFG $\{\Sigma, V, s, R\}$
by a recursive procedure:
(i) find the most frequent symbol pair $p$ from $T$, then
(ii) replace $p$ with a new variable $v$ then add $v$ to $V$ and $v\rightarrow p$ to $R$.
It terminates when the frequency of the most frequent pair is less than 2,
and adds $s\rightarrow T$ to $R$ at that time.
In general, at first,
\rp replaces each $a\in\Sigma$ with a new variable $v_a$,
and adds $v_a$ to $V$ and $v_a\rightarrow a$ to $R$.

{\em \mrrp} is a variant of \rp proposed by Furuya et al.~\cite{Furuya+:DCC2019},
that finds and replaces the most frequent conditional maximal repeat, 
instead of the most frequent pair.

\begin{lemma}[\cite{Furuya+:DCC2019}]\label{lem:mrrp}
	For a given text with length $n$,
	let us denote the grammar constructed by \mrrp by $\{\Sigma,V,s,R\}$.
	Then,
	\mrrp works in $\ord (n)$ expected time and 
	$5n + 4|\Sigma|^2 + 4V + \lceil \sqrt{n + 1} \rceil - 1$ words of space.
\end{lemma}

\subsection{Encoding methods}

For a given text,
{\em $i$-bit encoding} is an encoding method
that represents each symbol of the text by $i$ bits.
{\em Fixed bit length encoding (FBLE)} 
represents each symbol by
$\lceil\log{m}\rceil$ bits,
where $m$ is the value of the maximum symbol of the text.

{\em Huffman coding}~\cite{Huffman:IEEE1952} is 
a popular algorithm for compact bit encoding of text.
For a given text,
Huffman coding assigns a variable number of bits to symbols of the text
based on their frequencies,
and represents each symbol by the assigned number of bits.

{\em Gamma encoding} (also known as Elias gamma encoding)~\cite{Elias:IEEE1975} is
an encoding scheme for positive integers.
To encode a given number $n$, gamma encoding 
puts $\lfloor\log{n}\rfloor$ 0s and appends the binary form of $n$.

{\em Run-length encoding (RLE)} converts a given text to two sequences;
{\em symbol sequence} $S$ and {\em length sequence} $L$.
Assuming that the given text is $r_1 r_2 \cdots r_q$, where $r_i = {a_i}^{k_i}$ 
with $a_i\in\mathcal{A}$ and $k_i\in\mathbb{N}^{+}$ for $1\le i\le q$.
Then, the obtained $S$ and $L$ are denoted $a_1 a_2 \cdots a_q$ and
$k_1 k_2 \cdots k_q$, respectively.

\section{\rlmrrp}\label{sec:grammar}

Furuya et al.~\cite{Furuya+:DCC2019} recently proposed \mrrp as a variant of \rp based on maximal repeats.
They reported that \mrrp practically improves
the efficiency of compression of \rp with respect to grammar size, 
especially for repetitive data.
In this section, 
we extend \mrrp to run-length grammar compression schemes
and present a new variant of \rp called \rlmrrp.

\subsection{Algorithm}

Let $x$ be a symbol and $k$ be a natural number such that $k \ge 2$.
For text $x^k$, its most frequent maximal repeat is $x^2$.
Conversely, for a given text $T$, if its most frequent maximal repeat is $x^2$,
there is a possibility that $T$ contains long repetition $x^k$.
\rlmrrp searches run $x^k$ in $T$ and replaces it
if the most frequent maximal repeat is $x^2$.
Otherwise, it works similar to \mrrp.
We show the \rlmrrp algorithm in Algorithm~\ref{alg:rlmrrp}.
\begin{algorithm}[tb]
	\caption{\rlmrrp}\label{alg:rlmrrp}
	\begin{algorithmic}[1]
		\Require $T$
		\Ensure $G = \{V,\Sigma,s,R\}$
		\State Replace each $a \in \Sigma$ in $T$ with a new variable $v_a$, then add $v_a$ to $V$ and $v_a\rightarrow a$ to $R$.
		\Loop
		\State Find the most frequent maximal repeat $r$.
		\If{$\occ{T}{r}<2$}
		\State Add $s\rightarrow T$ to $R$.
		\State \Return $G$
		\EndIf
		\If{$r=x^2$ with variable $x$}
		\State Replace each run $x^k$ with a new variable $v_k$, \label{line:replace} then add $v_k$ to $V$ and $v_k\rightarrow x^k$ to $R$.
		\Else
		\If{$|r|>2$ and $r[1]=r[|r|]$}
		\State $r\leftarrow r[1..|r|-1]$
		\EndIf
		\State Replace each $r$ in $T$ with a new variable $v$, then add $v$ to $V$ and $v\rightarrow r$ to $R$.
		\EndIf
		\EndLoop
	\end{algorithmic}
\end{algorithm}

Let $g$ and $g_{\mathit{rl}}$ be the size of the grammar constructed by \mrrp and \rlmrrp, respectively.
It is easy to show that there is a case where $g_{\mathit{rl}}<g$,
since $g=\Theta(\log{n})$ and $g_{\mathit{rl}}=\Theta(1)$ holds for unary input text with length $n$.
However, note that $g_{\mathit{rl}}>g$ holds in some cases.
Assuming that the input text is $a^{2^1}ba^{2^2}b\cdots a^{2^m}b$,
then clearly, $g=\Theta(\log{m})$ and $g_{\mathit{rl}}=\Theta(m)$.

\subsection{Implementation}

We implement \rlmrrp by adding extra hash to the implementation of \mrrp.
The extra hash is used in the replacement phase (c.f.~Line~\ref{line:replace} in Algorithm~\ref{alg:rlmrrp})
to check whether the same run has occurred previously.
If the run has already occurred,
the same variable is used for the replacement.
Otherwise, a new variable is required.

\begin{theorem}\label{thm:time}
	\rlmrrp works in $\ord(n)$ expected time,
	where $n$ is the length of the input text.
\end{theorem}

\begin{proof}
	Compared with \mrrp,
	\rlmrrp requires an additional operation 
	when it replaces each run with a new variable,
	that is, it checks whether the same run has occurred previously.
	We assume that the extra hash works in $\ord(1)$ expected time;
	thus, the time complexity of \rlmrrp is equal to that of \mrrp,
	that is, by Lemma~\ref{lem:mrrp}, \rlmrrp works in $\ord(n)$ expected time.
\end{proof}

\begin{theorem}\label{thm:space}
	For a given text with length $n$,
	let us denote the grammar constructed by \rlmrrp by $\{\Sigma,V,s,R\}$.
	Then,
	\rlmrrp works in
	$6n + 4|\Sigma|^2 + 4V + \lceil \sqrt{n + 1} \rceil - 1$ words of space.
\end{theorem}

\begin{proof}
	Compared with \mrrp,
	\rlmrrp requires an additional space for the extra hash,
	which maintains the length of the runs that occur in the text.
	The total length of such runs is at most $n$;
	thus, the hash is at most $n$ words of space.
	Therefore, by Lemma~\ref{lem:mrrp}, 
	\rlmrrp works in 
	$6n + 4|\Sigma|^2 + 4V + \lceil \sqrt{n + 1} \rceil - 1$ words of space.
\end{proof}

\section{Bit encoding}\label{sec:encoding}

Let $G = \{\Sigma,V,s,R\}$ be a grammar constructed by \rp, \mrrp, or \rlmrrp,
where $\Sigma = \{a_1,\cdots,a_{\sigma}\}$,
$V = \{1,\cdots,(\sigma+d+1)\}$, $s=(\sigma+d+1)$, and
$R = \{1\rightarrow a_1,\cdots,\sigma\rightarrow a_{\sigma}, 
(\sigma+1)\rightarrow\alpha_1,\cdots,(\sigma+d)\rightarrow\alpha_{d},
(\sigma+d+1)\rightarrow \tau\}$.
In what follows,
for each run-length rule $v_i\rightarrow {v_j}^k$,
we write its right-hand side in the form of a symbol sequence $0kv_j$,
where $0$ is a special symbol
that implies that the expression is the right-hand-side of a run-length rule.
By this representation,
we treat RLCFG as CFG in what follows.

Compressed data is finally stored as bit sequences.
The simplest method for encoding $G$ is
by converting $G$ to a text and encoding the text using general text encoding schemes,
such as $i$-bit encoding, FBLE, or Huffman coding
(e.g., a \rp implementation by Navarro~\cite{Navarro:impl} uses 32-bit encoding).
For example, 
we can convert $G$ to
$a_1 \cdots a_{\sigma} \diamond \alpha_1 \diamond \alpha_2 \diamond \cdots \diamond \alpha_{d} \diamond \tau$
with a special symbol $\diamond$, called a {\em delimiter}.
Let $g$ denote the size of $G$,
and then the length of such text is 
$\sigma+1+\sum_i^d(|\alpha_i|+1)+|\tau| = g+d+1$.
Here, if the length of each $\alpha_i$ is 2,
the number of delimiters can be reduced, that is,
we can convert $G$ to
$a_1 \cdots a_{\sigma} \diamond \alpha_1 \alpha_2 \cdots \alpha_{d} \diamond \tau$
with length 
$\sigma+1+\sum_i^d|\alpha_i|+1+|\tau| = g+2$.
This implies that the final bit sequence of grammar for \rp can be smaller than
that of \mrrp or \rlmrrp,
even if the size of the grammar for \rp is larger than that of \mrrp or \rlmrrp.

\subsection{A previous effective method for \rp}

In another approach,
Bille et al.~\cite{Bille+:DCC2017,Prezza:github} proposed a variant of \rp and an effective encoding for it.
They partially sorted the rules of grammar and encoded the grammar by using
{\em packed gamma encoding (PGE)}, defined as follows.

\begin{Definition}[PGE]\label{def:PGE}
	Given a text $T$ and a natural number $\varepsilon$.
	Let $D$ be a sequence such that
	$\lceil\log{l_1}\rceil \lceil\log{l_2}\rceil$ $\cdots \lceil\log{l_q}\rceil$,
	where $l_i$ is the value of the maximum symbol
	in $T[j..j+\varepsilon-1]$ with $j=\varepsilon(i-1)$
	and $q=\lfloor|T|/\varepsilon\rfloor$,
	$D_{\mathit{delta}}$ be a sequence such that
	$D_{\mathit{delta}}[1] = D[1] + 1$ and $D_{\mathit{delta}}[i] = |D[i]-D[i-1]| + 1$ for $1 < i \le q$, and
	$D_{\mathit{pms}}$ be a bit sequence such that
	$D_{\mathit{pms}}[1] = 1$ and for $1 < i \le q$ if $D[i] \ge D[i-1]$
	then $D_{\mathit{pms}}[i]$ is 1 and otherwise 0.
	$S_1$ and $L_1$ denote the symbol sequence and the length sequence obtained by RLE of $D_{\mathit{delta}}$, respectively.
	Similarly, $S_2$ and $L_2$ denote those obtained by RLE of $L_1$.
	Then,
	{\em PGE} is an encoding method that
	encodes $T$ to a bit sequence consisting of the following five bit sequences.
	\begin{enumerate}
		\item A gamma-encoded bit sequence of $S_1$.
		\item A gamma-encoded bit sequence of $S_2$.
		\item A gamma-encoded bit sequence of $L_2$.
		\item A bit sequence obtained from $T$ by representing each symbol $T[i]$ 
			by $D[\lfloor i/\varepsilon\rfloor]$ bits for $1\le i\le |T|$.
		\item $D_{\mathit{pms}}$.
	\end{enumerate}
\end{Definition}
It is expected that PGE performs well for a type of texts
such that the values of their symbols are close to that of their adjacent symbols.
Bille et al.~\cite{Bille+:DCC2017,Prezza:github} applied PGE to \rp in the following manner;
(i) construct two texts $X$ and $X_{\mathit{delta}}$ such that
$X = \max(\alpha_i[1], \alpha_i[2])$ and
$X_{\mathit{delta}} = |\alpha_i[1] - \alpha_i[2]|$
(note that the length of each $\alpha_i$ is 2 in \rp),
(ii) construct a bit sequence $X_{\mathit{pms}}$ such that
if $X[i]$ is $\alpha_i[1]$ then $X_{\mathit{pms}}[i]$ is 1 and otherwise 0,
(iii) store $X_{\mathit{delta}}$, $X_{\mathit{pms}}$, and PGE encoded bit sequence of $X$.
This encoding scheme requires that the length of each $\alpha_i$ is 2;
thus, it can be applied it neither to \mrrp nor \rlmrrp.

\subsection{Encoding via post-order partial parse tree (POPPT)}\label{sec:POPPT}

A {\em partial parse tree}~\cite{Rytter:TCS2003} is an ordered tree
formed by traversing the parse tree in the depth-first manner and
pruning out all descendants under each node of variables appearing no less than twice.
A {\em POPPT}~\cite{Maruyama+:SPIRE2013} is 
a partial parse tree whose internal nodes contain post-order variables. 
A {\em post-order CFG (POCFG)}~\cite{Maruyama+:SPIRE2013} is a CFG whose partial parse tree is a POPPT.
For compact encoding of general grammars (which implies not only CFGs with $|\alpha_i|=2$ for each $i$),
succinct representation of POCFG is useful.

Takabatake et al.~\cite{Takabatake+:SPIRE2014}%
\footnote{In the definition in~\cite{Takabatake+:SPIRE2014}, 
they use a second bit sequence to mark the leaves of the partial parse tree
to use it as the rank/select dictionary.
But we omit it because we do not use the bit sequence in our method.
}
presented a method for encoding a POCFG to its succinct representation
comprising a bit sequence $B$ and a text $U$.
For a given POCFG,
$P$ denotes its partial parse tree (which is a POPPT).
Then,
$B$ is built by traversing $P$ in the post-order and 
putting $c$ 0s and one 1 for a node with $c$ children.
Finally, we put one $0$ in $B$ to represent the super node.
$U$ stores symbols of leaves of $P$ left to right.

In the previous study undertaken by Takabatake et al.~\cite{Takabatake+:SPIRE2014},
Maruyama et al.~\cite{Maruyama+:SPIRE2013} proposed
a similar encoding method for SLPs.
The method constructs a bit sequence $B$ and a text $U$, in a similar manner,
but $B$ is slightly more effective.
$B$ is built by traversing the POPPT in the post-order, and
putting one 0 in $B$ if the node is a leaf, or one 1 otherwise.
Note that we can easily apply this method to the constructed grammar by \rp with
decomposition $\tau$ to 
$\{(\sigma+d+1)\rightarrow\tau[1]\tau[2],(\sigma+d+2)\rightarrow(\sigma+d+1)\tau[3],\cdots,
(\sigma+d+|\tau|-1)\rightarrow(\sigma+d+|\tau|-2)\tau[|\tau|]\}$.
By this decomposition, the size of the grammar increases by $|\tau|$.
However, it does not affect the final representation
since the variables $(\sigma+d+1),\cdots,(\sigma+d+|\tau|-1)$ do not explicitly appear
and only symbols that occur in $\tau$ are put in $U$.

\subsection{Combination of POPPT and PGE}

In both of the methods stated in Sec.~\ref{sec:POPPT},
we finally encode $U$ to a bit sequence.
In previous methods,
$U$ is encoded 
by representing each symbol $U[i]$ by $\lceil\log{(i+|\Sigma|)}\rceil$ bits for $1\le i\le |U|$.
This method is referred to as {\em increasing bit length encoding (IBLE)}.
Here,
note that $U[i]\le i+|\Sigma|$ holds 
since in POPPT, the value of a leaf node is 
at most the number of internal nodes in post-order until the leaf node.

As another method for encoding $U$,
we propose a scheme that uses PGE.
It is expected that PGE performs well for encoding $U$
since
there is a tendency that the values of symbols in $U$ are close to that of their adjacent symbols.

\section{Experiments}\label{sec:experiments}

We implemented \rp, \mrrp, and \rlmrrp
and conducted experiments
to compare their performances.
The experiments involved measuring the sizes of grammars,
the execution time and the peak memory usage for grammar construction,
the sizes of the final compressed files,
and the execution time for encoding the grammars.
We used the datasets listed in Table~\ref{table:datasets} for testing.
All datasets were obtained from Repetitive Corpus produced in Pizza\&Chili Corpus~\cite{pizza_chili}.
All tests were conducted on Intel(R) Core i7-7800X 3.50GHz 12core with 64GB RAM.
The OS was Linux (Ubuntu 16.04.2, 64bit) running kernel 4.15.0.
All the programs were compiled by \texttt{rustc} version 1.35.0 with \texttt{--release} option.

\begin{table*}[t]
\begin{center}
\caption{
\small
	Datasets used in experiments.
	Here, $|\Sigma|$ is the alphabet size, representing the number of kinds of symbols that occur in each dataset.
	Type implies classification by the scheme used to generate the dataset;
	artificially created symbol sequences (A), artificially generated by adding repetitiveness to real data (PR), and
	real repetitive data (R).
}
\vspace{1ex}
\label{table:datasets}
\begin{tabular}{l@{\quad}r@{\quad}r@{\quad}r@{\quad}l}
	\hline
	\multicolumn{1}{l}{Name} & \multicolumn{1}{r}{Size (bytes)} & \multicolumn{1}{r}{$|\Sigma|$} 
	& \multicolumn{1}{r}{Type} & \multicolumn{1}{l}{Description}                      \\
	\hline
	fib41           & 267,914,296   & 2     & A   & ~~Fibonacci string \\
	dna.001.1       & 104,857,600   & 5     & PR  & ~~100 $\times$ 1MiB prefix of human genome \\
	sources.001.2   & 104,857,600   & 98    & PR  & ~~100 $\times$ 1MiB prefix of Linux and GCC sources \\
	coreutils       & 205,281,778   & 236   & R   & ~~9 versions of GNU Coreutils source \\
	einstein.en.txt & 467,626,544   & 139   & R   & ~~Edit history of Wikipedia for Albert Einstein \\
	influenza       & 154,808,555   & 15    & R   & ~~78,041 DNA sequences of Haemophilus Influenzae \\
	para            & 429,265,758   & 5     & R   & ~~36 DNA sequences of Saccharomyces Paradoxus \\
	world\_leaders  & 46,968,181    & 89    & R   & ~~CIA World Leaders from Jan. 2003 to Dec. 2009 \\
	\hline
\end{tabular}
\end{center}
\end{table*}

\subsection{Grammar construction}

Table~\ref{table:grammar} shows the sizes of grammars constructed 
by \rp, \rpps, \mrrp and \rlmrrp,
the execution time, and the peak memory usage for grammar construction.
\rpps is a variant of \rp stated in~\cite{Bille+:DCC2017}\footnotemark[2],
which partially sorts the rules of grammar in the procedure.
In any case,
\mrrp and \rlmrrp outperform both \rp and \rpps,
while it was stated in~\cite{Furuya+:DCC2019} that
there is a case in which \mrrp is theoretically inferior to \rp.
In particular,
for all datasets except \texttt{coreutils},
\rlmrrp constructs the smallest grammars.
In addition, as shown in the table,
\rlmrrp runs faster than \mrrp in all cases except sources.001.2.

\begin{table*}[t]
\begin{center}
\caption{
\small
	The sizes of the generated grammars,
	the execution time, and
	the peak memory usage for grammar construction.
	Let $G = \{\Sigma,V,s,R\}$ be a constructed grammar,
	where $\Sigma = \{a_1,\cdots,a_{\sigma}\}$,
	$V = \{1,\cdots,(\sigma+d+1)\}$, $s=(\sigma+d+1)$, and
	$R = \{1\rightarrow a_1,\cdots,\sigma\rightarrow a_{\sigma}, 
	(\sigma+1)\rightarrow\alpha_1,\cdots,(\sigma+d)\rightarrow\alpha_{d},
	(\sigma+d+1)\rightarrow \tau\}$.
	Then, from the top row,
	each cell in the table represents $d$, $(\sum_{i=0}^{d}|\alpha_i|)$, $|\tau|$, and the size of the grammar for $G$.
	The fifth and sixth rows separated by a line represent 
	the average running time of five executions with seconds and
	the average peak memory usage of five executions with kilobytes (kB), respectively.
}
\vspace{2ex}
\label{table:grammar}
\fontsize{9}{11}\selectfont
\begin{tabular}{l|rrrr} \hline
	& \multicolumn{1}{l}{\rp} & \multicolumn{1}{l}{\rpps} 
	& \multicolumn{1}{l}{\mrrp} & \multicolumn{1}{l}{\rlmrrp} \\ \hline
fib41			& \textbf{38} & \textbf{38} & \textbf{38} & \textbf{38} \\
				& \textbf{76} & \textbf{76} & \textbf{76} & \textbf{76} \\
				& \textbf{ 3} & \textbf{ 3} & \textbf{ 3} & \textbf{ 3} \\
				& \textbf{81} & \textbf{81} & \textbf{81} & \textbf{81} \\
\cline{2-5}		& \textbf{67.163}	& 67.183	& 81.341	& 81.162	\\
				& \textbf{18,122,200}	& 18,122,220	& 18,921,276	& 18,921,412	\\
\hline
dna.001.1		& 261,023	& 261,239	& 223,983	& \textbf{223,612}	\\
				& 522,046	& 522,478	& 485,514	& \textbf{485,251}	\\
				& 498,612	& 498,402	& 496,566	& \textbf{494,406}	\\
				& 1,020,663	& 1,020,885	& 982,085	& \textbf{979,662}	\\
\cline{2-5}		& \textbf{61.787}	& 61.798	& 70.493	& 67.244	\\
				& 7,684,652	& 7,685,988	& \textbf{7,660,588}	& 7,660,980	\\
\hline                                                          
sources.001.2   & 709,174	& 709,052	& 400,258	& \textbf{400,213  }	\\
				& 1,418,348	& 1,418,104	& 1,109,686	& \textbf{1,109,548}	\\
				& 183,656	& 183,583	& 181,393	& \textbf{181,253  }	\\
				& 1,602,102	& 1,601,785	& 1,291,177	& \textbf{1,290,899}	\\
\cline{2-5}		& \textbf{64.770}	& 64.924	& 68.147	& 69.077	\\
				& 7,708,468	& 7,708,868	& 7,620,056	& \textbf{7,619,772}	\\
\hline
coreutils       & 1,833,094	& 1,833,918	& 436,515	& \textbf{436,443}	\\
				& 3,666,188	& 3,667,836	& \textbf{2,269,133}	& 2,269,393	\\
				& 154,036	& 154,001	& 153,622	& \textbf{153,611}	\\
				& 3,820,460	& 3,822,073	& \textbf{2,422,991}	& 2,423,240	\\
\cline{2-5}		& \textbf{122.824}	& 124.309	& 137.012	& 128.219	\\
				& 15,529,364	& 15,529,764	& 15,226,140	& \textbf{15,225,932}	\\
\hline
einstein.en.txt & 100,681	& 100,641	& 49,373	& \textbf{49,221 }	\\
				& 201,362	& 201,282	& \textbf{150,105}	& 150,173	\\
				& 62,492	& 62,580	& 62,318	& \textbf{62,096 }	\\
				& 263,993	& 264,001	& 212,562	& \textbf{212,408}	\\
\cline{2-5}		& 294.028	& \textbf{293.325}	& 323.260	& 320.619	\\
				& 25,181,396	& 25,181,724	& 24,741,612	& \textbf{24,735,416}	\\
\hline
influenza       & 659,560	& 659,473	& 427,595	& \textbf{423,419  }	\\
				& 1,319,120	& 1,318,946	& 1,088,157	& \textbf{1,077,405}	\\
				& 897,431	& 898,010	& 894,544	& \textbf{887,131  }	\\
				& 2,216,566	& 2,216,971	& 1,982,716	& \textbf{1,964,551}	\\
\cline{2-5}		& \textbf{87.705}	& 87.819	& 103.473	& 98.723	\\
				& 13,240,848	& 13,242,928	& 13,228,356	& \textbf{13,109,972}	\\
\hline
para			& 3,076,152	& 3,077,085	& \textbf{1,079,287}	& 1,082,467	\\
				& 6,152,304	& 6,154,170	& 4,157,167	& \textbf{4,145,790}	\\
				& 1,142,696	& 1,142,356	& 1,134,361	& \textbf{1,121,371}	\\
				& 7,295,005	& 7,296,531	& 5,291,533	& \textbf{5,267,166}	\\
\cline{2-5}		& \textbf{248.267}	& 249.652	& 279.212	& 260.333	\\
				& 32,160,672	& 32,165,360	& \textbf{31,602,712}	& 31,603,136	\\
\hline
world\_leaders  & 209,071	& 209,079	& 99,910	& \textbf{98,078 }	\\
				& 418,142	& 418,158	& 309,031	& \textbf{306,091}	\\
				& 98,127	& 98,210	& 97,712	& \textbf{94,851 }	\\
				& 516,358	& 516,457	& 406,832	& \textbf{401,031}	\\
\cline{2-5}		& 20.214	& 20.427	& 23.553	& \textbf{18.586}	\\
				& 4,222,924	& 4,222,713	& \textbf{4,163,532}	& 4,164,368	\\
\hline
\end{tabular}
\end{center}
\end{table*}

\footnotetext[2]{
	While a program of the algorithm implemented by the authors is available in~\cite{Prezza:github},
	we implemented by ourselves and used for testing
	for fairness of comparison
	(the implementation of~\cite{Prezza:github} mainly aims to reduce the memory usage
	and the performance of the execution time slightly decreases for the purpose).
}

\subsection{Encoding the grammars}

Tables~\ref{table:encoding_rp}, \ref{table:encoding_mrrp} and \ref{table:encoding_rlmrrp} 
show the sizes of files compressed by \rp, \mrrp, and \rlmrrp, respectively.
We tested 8 encoding methods for \rp
and 6 encoding methods for \mrrp and \rlmrrp.
Each table shows the execution time of each encoding method.

\begin{table*}[t]
\begin{center}
\vspace{-6ex}
\caption{
\small
	The sizes of files compressed by \rp and the execution time for encoding the grammars.
	In each cell, the first row represents the size (bytes),
	whereas the second row with parentheses represents the compression ratio
	(compressed file size)$/$(input file size)$\times 100$~($\%$).
	The third row separated by a line represents the average running time of five executions with seconds.
}
\label{table:encoding_rp}
\footnotesize
\begin{tabular}{l|rrrrrrrr} \hline
& \multicolumn{8}{l}{\rp} \\ \cline{2-9}
	& \multicolumn{1}{l}{\texttt{32bit}} & \multicolumn{1}{l}{\texttt{fble}} & \multicolumn{1}{l}{\texttt{huffman}} 
	& \multicolumn{1}{l}{\texttt{ps+pge6}} & \multicolumn{1}{l}{\texttt{ps+pge8}} 
	& \multicolumn{1}{l}{\texttt{poppt+ible}} 
	& \multicolumn{1}{l}{\texttt{poppt+pge6}} & \multicolumn{1}{l}{\texttt{poppt+pge8}} \\ \hline
fib41			& 327 & 85 & 251 & 83 & 84 & \textbf{50} & 69 & 71 \\
				& (0.0000) & (0.0000) & (0.0000) &	(0.0000) &	(0.0000) &	\textbf{(0.0000)} &	(0.0000) & (0.0000) \\ 
\cline{2-9}		& \textbf{0.000}	& \textbf{0.000}	& \textbf{0.000}	& \textbf{0.000}	& \textbf{0.000}	& \textbf{0.000}	& \textbf{0.000}	& \textbf{0.000}	\\
\hline
dna				& 4,082,646	& 2,296,506	& 3,211,576	& 1,783,251	& \textbf{1,778,453}	& 1,957,954	& 1,906,365	& 1,911,272	\\
~~.001.1		& (3.8935)	& (2.1901)	& (3.0628)	& (1.7006)	& \textbf{(1.6961)}	& (1.8673)	& (1.8181)	& (1.8277)	\\
\cline{2-9}		& 0.085	& 0.119	& 0.627	& \textbf{0.078}	& \textbf{0.078}	& 0.156	& 0.166	& 0.162	\\
\hline
sources			& 6,408,123	& 4,005,269	& 6,713,711	& 2,606,298	& 2,604,665	& \textbf{2,324,485}	& 2,342,488	& 2,334,697	\\
~~.001.2		& (6.1113)	& (3.8197)	& (6.4027)	& (2.4856)	& (2.4840)	& \textbf{(2.2168)}	& (2.2340)	& (2.2265)	\\
\cline{2-9}		& \textbf{0.148}	& 0.214	& 1.438	& 0.169	& 0.167	& 0.272	& 0.286	& 0.283	\\
\hline
coreutils       & 15,281,141 	& 10,028,722	& 17,020,682	& 5,655,042	& 5,657,054	& \textbf{5,451,520}	& 5,469,724	& 5,461,667	\\
				& (7.4440)	& (4.8853)	& (8.2914)	& (2.7548)	& (2.7558)	& \textbf{(2.6556)}	& (2.6645)	& (2.6606)	\\
\cline{2-9}		& \textbf{0.329}	& 0.533	& 3.867	& 0.420	& 0.394	& 0.688	& 0.714	& 0.710	\\
\hline
einstein		& 1,055,564	& 560,999	& 938,266	& 439,083	& 441,650	& 375,523	& \textbf{374,902}	& 374,938	\\
~~.en.txt		& (0.2257)	& (0.1200)	& (0.2006)	& (0.0939)	& (0.0944)	& (0.0803)	& \textbf{(0.0802)}	& (0.0802)	\\
\cline{2-9}		& \textbf{0.024}	& 0.031	& 0.165	& 0.027	& 0.027	& 0.037	& 0.039	& 0.038	\\
\hline
influenza       & 8,866,228	& 5,541,429	& 7,813,734	& 4,140,255	& \textbf{4,137,727}	& 4,214,266	& 4,201,129	& 4,198,308	\\
				& (5.7272)	& (3.5795)	& (5.0474)	& (2.6744)	& \textbf{(2.6728)}	& (2.7222)	& (2.7138)	& (2.7119)	\\
\cline{2-9}		& \textbf{0.189}	& 0.295	& 1.603	& 0.203	& 0.199	& 0.387	& 0.409	& 0.402	\\
\hline
para			& 29,180,014	& 20,061,278	& 30,566,314	& 11,812,763	& 11,893,263	& 12,135,356	& \textbf{11,710,363}	& 11,759,392	\\
				& (6.7977)	& (4.6734)	& (7.1206)	& (2.7519)	& (2.7498)	& (2.8270)	& \textbf{(2.7280)}	& (2.7394)	\\
\cline{2-9}		& \textbf{0.631}	& 1.053	& 7.193	& 0.723	& 0.712	& 1.423	& 1.460	& 1.455	\\
\hline
world\_			& 2,065,174	& 1,161,820	& 1,953,121	& 796,540	& 796,666	& 741,111	& 740,316	& \textbf{739,570}	\\
	~~leaders		& (4.3970)	& (2.4736)	& (4.1584)	& (1.6959)	& (1.6962)	& (1.5779)	& (1.5762)	& \textbf{(1.5746)}	\\
\cline{2-9}		& \textbf{\textbf{0.046}}	& 0.061	& 0.380	& 0.049	& 0.049	& 0.078	& 0.081	& 0.081	\\
\hline
\end{tabular}
\end{center}
\end{table*}

\begin{table*}[t]
\begin{center}
\vspace{-6ex}
\caption{
\small
	The sizes of files compressed by \mrrp and the execution time for encoding the grammars.
	In each cell, the first row represents the size (bytes),
	whereas the second row with parentheses represents the compression ratio
	(compressed file size)$/$(input file size)$\times 100$~($\%$).
	The third row separated by a line represents the average running time of five executions with seconds.
}
\vspace{1ex}
\label{table:encoding_mrrp}
\footnotesize
\begin{tabular}{l|rrrrrr} \hline
& \multicolumn{6}{l}{\mrrp} \\ \cline{2-7}
	& \multicolumn{1}{l}{\texttt{32bit}} & \multicolumn{1}{l}{\texttt{fble}} & \multicolumn{1}{l}{\texttt{huffman}} 
	& \multicolumn{1}{l}{\texttt{poppt+ible}} 
	& \multicolumn{1}{l}{\texttt{poppt+pge6}} & \multicolumn{1}{l}{\texttt{poppt+pge8}} \\ \hline
fib41			& 429 & 118 & 264 & \textbf{60} & 79 & 78 \\
				& (0.0000) & (0.0000) & (0.0000) & \textbf{(0.0000)} & (0.0000) & (0.0000) \\
\cline{2-7}		& \textbf{0.000}	& \textbf{0.000}	& \textbf{0.000}	& \textbf{0.000}	& \textbf{0.000}	& \textbf{0.000}	\\
\hline
dna.001.1		& 4,824,266	& 2,713,667	& 3,073,145	& 1,918,499	& 1,895,294	& \textbf{1,894,870}	\\
				& (4.6008)	& (2.5880)	& (2.9308)	& (1.8296)	& (1.8075)	& \textbf{(1.8071)}	\\
\cline{2-7}		& \textbf{0.096}	& 0.139	& 0.560	& 0.144	& 0.151	& 0.149	\\
\hline
sources.001.2   & 6,765,455	& 4,017,172	& 4,791,544	& 2,373,197	& 2,343,157	& \textbf{2,335,164}	\\
				& (6.4520)	& (3.8311)	& (4.5696)	& (2.2633)	& (2.2346)	& \textbf{(2.2270)}	\\
\cline{2-7}		& \textbf{0.145}	& 0.217	& 0.933	& 0.215	& 0.219	& 0.217	\\
\hline
coreutils       & 11,437,333 	& 6,791,346	& 7,368,357	& 5,258,079	& 5,115,689	& \textbf{5,106,577}	\\
				& (5.5715)	& (3.3083)	& (3.5894)	& (2.5614)	& (2.4920)	& \textbf{(2.4876)}	\\
\cline{2-7}		& \textbf{0.222}	& 0.363	& 1.427	& 0.393	& 0.403	& 0.394	\\
\hline
einstein.en.txt & 1,047,332	& 523,884	& 626,349	& 371,338	& 363,071	& \textbf{362,624}	\\
			& (0.2240)	& (0.1120)	& (0.1339)	& (0.0794)	& (0.0776)	& \textbf{(0.0775)}	\\
\cline{2-7}		& \textbf{0.021}	& 0.027	& 0.096	& 0.026	& 0.028	& 0.028	\\
\hline
influenza       & 9,641,208	& 5,724,503	& 6,427,285	& 4,123,574	& 4,071,746	& \textbf{4,064,247}	\\
				& (6.2278)	& (3.6978)	& (4.1518)	& (2.6637)	& (2.6302)	& \textbf{(2.6253)}	\\
\cline{2-7}		& \textbf{0.198}	& 0.313	& 1.253	& 0.326	& 0.346	& 0.351	\\
\hline
para			& 25,483,274	& 16,723,417	& 16,887,956	& 12,117,901	& \textbf{11,269,822}	& 11,306,815	\\
				& (5.9365)	& (3.8958)	& (3.9341)	& (2.8229)	& \textbf{(2.6254)}	& (2.6340)	\\
\cline{2-7}		& \textbf{0.502}	& 0.884	& 3.565	& 1.009	& 1.019	& 1.005	\\
\hline
world\_leaders  & 2,026,710	& 1,076,841	& 1,275,874	& 737,552	& 719,313	& \textbf{717,965}	\\
				& (4.3151)	& (2.2927)	& (2.7165)	& (1.5703)	& (1.5315)	& \textbf{(1.5286)}	\\
\cline{2-7}		& \textbf{0.041}	& 0.057	& 0.226	& 0.058	& 0.056	& 0.056	\\
\hline
\end{tabular}
\end{center}
\end{table*}

\begin{table*}[t]
\begin{center}
\vspace{-6ex}
\caption{
\small
	The sizes of files compressed by \rlmrrp and the execution time for encoding the grammars.
	In each cell, the first row represents the size~(bytes),
	whereas the second row with parentheses represents the compression ratio
	(compressed file size)$/$(input file size)$\times 100$~($\%$).
	The third row separated by a line represents the average running time of five executions with seconds.
}
\vspace{1ex}
\label{table:encoding_rlmrrp}
\footnotesize
\begin{tabular}{l|rrrrrr} \hline
& \multicolumn{6}{l}{\rlmrrp} \\ \cline{2-7}
	& \multicolumn{1}{l}{\texttt{32bit}} & \multicolumn{1}{l}{\texttt{fble}} & \multicolumn{1}{l}{\texttt{huffman}} 
	& \multicolumn{1}{l}{\texttt{poppt+ible}} 
	& \multicolumn{1}{l}{\texttt{poppt+pge6}} & \multicolumn{1}{l}{\texttt{poppt+pge8}} \\ \hline
fib41			& 479 & 118 & 264 & \textbf{60} & 79 & 78 \\
				& (0.0000) & (0.0000) & (0.0000) & \textbf{(0.0000)} & (0.0000) & (0.0000) \\
\cline{2-7}		& \textbf{0.000}	& \textbf{0.000}	& \textbf{0.000}	& \textbf{0.000}	& \textbf{0.000}	& \textbf{0.000}	\\
\hline
dna.001.1		& 4,813,090	& 2,707,381	& 3,067,205	& 1,913,276	& 1,889,730	& \textbf{1,889,630}	\\
				& (4.5901)	& (2.5820)	& (2.9251)	& (1.8246)	& (1.8022)	& \textbf{(1.8021)}	\\
\cline{2-7}		& \textbf{0.096}	& 0.144	& 0.551	& 0.144	& 0.150	& 0.150	\\
\hline
sources.001.2   & 6,764,163	& 4,016,405	& 4,789,737	& 2,372,574	& 2,342,202	& \textbf{2,334,317}	\\
				& (6.4508)	& (3.8303)	& (4.5678)	& (2.2627)	& (2.2337)	& \textbf{(2.2262)}	\\
\cline{2-7}		& \textbf{0.141}	& 0.215	& 0.929	& 0.208	& 0.217	& 0.216	\\
\hline
coreutils       & 11,438,025 	& 6,791,756	& 7,367,012	& 5,258,904	& 5,115,868	& \textbf{5,106,824}	\\
				& (5.5719)	& (3.3085)	& (3.5887)	& (2.5618)	& (2.4921)	& \textbf{(2.4877)}	\\
\cline{2-7}		& \textbf{0.222}	& 0.359	& 1.330	& 0.381	& 0.395	& 0.393	\\
\hline
einstein.en.txt & 1,046,112	& 523,274	& 624,555	& 371,298	& 362,760	& \textbf{362,372}	\\
			& (0.2237)	& (0.1119)	& (0.1336)	& (0.0794)	& (0.0776)	& \textbf{(0.0775)}	\\
\cline{2-7}		& \textbf{0.022}	& 0.027	& 0.097	& 0.026	& 0.028	& 0.028	\\
\hline
influenza       & 9,551,844	& 5,671,443	& 6,367,819	& 4,085,808	& 4,033,158	& \textbf{4,025,295}	\\
				& (6.1701)	& (3.6635)	& (4.1134)	& (2.6393)	& (2.6053)	& \textbf{(2.6002)}	\\
\cline{2-7}		& \textbf{0.192}	& 0.301	& 1.223	& 0.318	& 0.334	& 0.329	\\
\hline
para			& 25,398,526	& 16,667,801	& 16,876,589	& 12,039,499	& \textbf{11,203,814}	& 11,240,382	\\
				& (5.9167)	& (3.8829)	& (3.9315)	& (2.8047)	& \textbf{(2.6100)}	& (2.6185)	\\
\cline{2-7}		& \textbf{0.501}	& 0.884	& 3.364	& 0.933	& 0.958	& 0.952	\\
\hline
world\_leaders  & 1,996,178	& 1,060,621	& 1,253,963	& 727,668	& 708,711	& \textbf{707,450}	\\
				& (4.2501)	& (2.2582)	& (2.6698)	& (1.5493)	& (1.5089)	& \textbf{(1.5062)}	\\
\cline{2-7}		& \textbf{0.040}	& 0.057	& 0.200	& 0.055	& 0.056	& 0.056	\\
\hline
\end{tabular}
\end{center}
\end{table*}

\texttt{32bit}, \texttt{fble}, and \texttt{huffman} 
convert a given grammar to text using the procedure introduced in the exordium of Sec~\ref{sec:encoding}
and encode the text by 32-bit encoding, FBLE, and Huffman coding, respectively.
Note that in \rp the number of delimiters is reduced in the converted text
compared with the others.
\texttt{32bit} is the simplest and it reflects the size of the grammar directly.
However, as shown in the tables,
32bit is too large for representing symbols.
In addition,
we observed that \texttt{huffman} tends to be large than \texttt{fble}.
We consider that 
Huffman coding does not perform well
because there are few symbols that occur repeatedly in the converted text.

For comparison, we implemented the encoding method proposed by Bille et al.~\cite{Bille+:DCC2017} 
for \rp (more precisely, for \rpps).
The implementation~\cite{Prezza:github} sets
the constant $\varepsilon$ as 6 (see Def.~\ref{def:PGE}).
However, we found that setting $\varepsilon = 8$ improves the efficiency of the compression for some cases;
thus, we tested for $\varepsilon = 6$ (\texttt{ps+pge6}) and $\varepsilon = 8$ (\texttt{ps+pge8}).
As shown in Table~\ref{table:encoding_rp},
both \texttt{ps+pge6} and \texttt{ps+pge8} significantly improve the efficiency of the compression
compared with methods that convert a given grammar to a text.

Finally,
we tested encoding methods using POPPT for every \rp variant.
For \rp,
we used the method proposed by Maruyama et al.~\cite{Maruyama+:SPIRE2013}
for constructing POPPT.
For \mrrp and \rlmrrp,
we adopted the method proposed by Takabatake et al.~\cite{Takabatake+:SPIRE2014}.
We represented the text in the succinct representation of POCFG in three ways;
using IBLE (\texttt{poppt+ible}),
PGE with $\varepsilon = 6$ (\texttt{poppt+pge6}),
and PGE with $\varepsilon = 8$ (\texttt{poppt+pge8}).
For both \mrrp and \rlmrrp,
the methods using POPPT show high compression efficiency,
particularly, \texttt{poppt+pge8}, which achieved the best compression ratio
in all cases, except fib41 and para
(in para, \texttt{poppt+pge6} was the best
whereas \texttt{poppt+pge8} was second best).
In addition, as shown in Table~\ref{table:encoding_rp},
methods using POPPT are effective even for grammars constructed by \rp.

From the above experiments,
we summarize the result of the best compression performance achieved by each \rp variant
in Table~\ref{table:best}.
For comparison, we also show the compression results by two famous file compressor, 
gzip (version 1.6, with \texttt{-9} option) and bzip2 (version 1.0.6, with \texttt{-9} option) in the table.
As the table shows,
\rp is effective
for artificial datasets (A) and pseudo-real datasets (PR).
In contrast,
for real datasets (R),
\mrrp and \rlmrrp show good performances,
whereas \rlmrrp improves the efficiency of the compression of \mrrp in all cases except coreutils.

\begin{table*}[t]
\begin{center}
\caption{
\small
	The sizes of files compressed by gzip, bzip2, and \rp variants 
	with their best encoding methods with regard to compression performance.
	From the top of the row,
	each cell represents the size~(bytes),
	the compression ratio
	(compressed file size)$/$(input file size)$\times 100$~($\%$),
	and the encoding method.
}
\vspace{1ex}
\label{table:best}
\footnotesize
\begin{tabular}{l|rrrrr} \hline
	& \multicolumn{1}{l}{gzip} & \multicolumn{1}{l}{bzip2}
	& \multicolumn{1}{l}{\rp} & \multicolumn{1}{l}{\mrrp} & \multicolumn{1}{l}{\rlmrrp} \\ \hline
fib41			& 1,176,257	& 14,893	& \textbf{50} & 60 & 60 \\
				& (0.4390)	& (0.0056)	& \textbf{(0.0000)} & (0.0000) & (0.0000) \\
				&			&			& \texttt{poppt+ible} & \texttt{poppt+ible} & \texttt{poppt+ible} \\
\hline
dna.001.1		& 28,486,029	& 27,385,893	& \textbf{1,778,453}	& 1,894,870	& 1,889,630	\\
				& (27.1664)	& (26.1172)	& \textbf{(1.6961)}	& (1.8071)	& (1.8021)	\\
				&	&	& \texttt{ps+pge8} & \texttt{poppt+pge8} & \texttt{poppt+pge8} \\
\hline
sources.001.2   & 36,023,271	& 34,619,138	& \textbf{2,324,485}	& 2,335,164	& 2,334,317 \\
				& (34.3545)	& (33.0154)	& \textbf{(2.2168)}	& (2.2270)	& (2.2262)	\\
				&	&	& \texttt{poppt+ible} & \texttt{poppt+pge8} & \texttt{poppt+pge8} \\
\hline
coreutils       & 49,920,838	& 32,892,028	& 5,451,520	& \textbf{5,106,577}	& 5,106,824 \\
				& (24.3182)	& (16.0229)	& (2.6556)	& \textbf{(2.4876)}	& (2.4877)	\\
				&	&	& \texttt{poppt+ible} & \texttt{poppt+pge8} & \texttt{poppt+pge8} \\
\hline
einstein.en.txt & 163,664,285	& 24,157,362	& 374,902	& 362,624	& \textbf{362,372 }  \\
				& (34.9989)	& (5.1660)	& (0.0802)	& (0.0775)	& \textbf{(0.0775)}	\\
				&	&	& \texttt{poppt+pge6} & \texttt{poppt+pge8} & \texttt{poppt+pge8} \\
\hline
influenza       & 10,636,889	& 10,197,176	& 4,137,727	& 4,064,247	& \textbf{4,025,295}	\\
				& (6.8710)	& (6.5870)	& (2.6728)	& (2.6253)	& \textbf{(2.6002) }	\\
				&	&	& \texttt{ps+pge8} & \texttt{poppt+pge8} & \texttt{poppt+pge8} \\
\hline
para			& 116,073,220	& 112,233,085	& 11,710,363	& 11,269,822	& \textbf{11,203,814}	\\
				& (27.0399)	& (26.1454)	& (2.7280)	& (2.6254)	& \textbf{(2.6100)}	\\
				&	&	& \texttt{poppt+pge6} & \texttt{poppt+pge6} & \texttt{poppt+pge6} \\
\hline
world\_leaders  & 8,287,665	& 3,260,930	& 739,570	& 717,965	& \textbf{707,450 } \\
				& (17.6453)	& (6.9428)	& (1.5746)	& (1.5286)	& \textbf{(1.5062)}	\\
				&	&	& \texttt{poppt+pge6} & \texttt{poppt+pge8} & \texttt{poppt+pge8} \\
\hline
\end{tabular}
\end{center}
\end{table*}

\section{Conclusion}\label{sec:conclusion}

In this study,
we extended \mrrp to RLCFG and designed a novel variant called \rlmrrp.
In addition, we proposed an encoding scheme for \mrrp and \rlmrrp
and conducted experiments
to compare their performances.
The experimental results show that
in practice \rlmrrp and the proposed encoding scheme achieve high compression performance
for real repetitive datasets.

As stated in Sec.~\ref{sec:introduction},
\rp practically achieves a high compression ratio compared with
the existing grammar compression methods.
However, it requires a large space for working;
this is also the same for \mrrp and \rlmrrp.
For \rp,
reducing the working space has recently been addressed~\cite{Bille+:DCC2017,Sakai+:DCC2019}.
Our future study will explore the development of
space efficient \mrrp/\rlmrrp algorithms.

\clearpage
\bibliographystyle{plain}
\bibliography{reference}

\end{document}